\documentclass[12pt,english,numsec]{extarticle}

\usepackage[T1]{fontenc}
\usepackage[latin9]{inputenc}
\usepackage[a4paper]{geometry}
\usepackage{amsmath}
\usepackage{amssymb}
\usepackage{amsfonts}
\usepackage{amsthm}
    \newtheorem{theorem}{Proposition}[section]

    \newtheorem{remark}{Remark}[section]
    \theoremstyle{definition}
    \newtheorem{definition}{Definition}[section]
\usepackage{mathrsfs}
\usepackage{tocbibind}
\usepackage{dsfont}
\usepackage{braket}
\usepackage{commath}
\usepackage{ragged2e}
\usepackage{bigints}
\usepackage{babel}
\usepackage{hyperref}
    \hypersetup{colorlinks=true, linktocpage=true, pageanchor=true, allcolors=black}

\makeatletter
\date{}
\makeatother

\begin{document}

\title{One Dimensional Fixed Point Interactions and the Resolvent Algebra - Simple Remarks}
\author{ Antonio Moscato\footnote{antonio.moscato@unitn.it}\\
\small{Dipartimento di Matematica, Universit{\'a} di Trento}\\
\small{Via Sommarive 14, I-38123 Povo (Trento), Italy} }

\maketitle

\begin{abstract}
\noindent In this paper, the resolvent algebra $\mathcal{R} \left( \mathbb{R}^2,\sigma \right)$ stability under dynamics induced by the symbolic Hamiltonians
\begin{itemize}
    \item $H = - \frac{d^2}{dx^2} + \alpha \delta\left( x - x_0 \right)$, $\alpha \in \mathbb{R}\setminus\{0\}, \, x_0 \in \mathbb{R}$,
    \item $H = - \frac{d^2}{dx^2} + \sum_{i = 1}^N \, \alpha_i \delta \left( x - x_i \right)$, $\alpha_i \in \mathbb{R} \setminus \{0\}$, $x_i \in \mathbb{R}: \, x_i \neq x_j$, $i,j \in \{1,\ldots,N\}$,
    \item $H = - \frac{d^2}{dx^2} + \sum_{i = 1}^{\infty} \alpha_i \delta \left( x - x_i \right)$, $\left\{ \alpha_i \right\}_i \in l^{1}\left( \mathbb{N} \right)$, $ x_i \in \mathbb{R}: \, x_i \neq x_j, \, i,j \in \mathbb{N}$
\end{itemize}
\noindent is proved: if $\pi_S$ is the Sch{\"o}dinger representation of $\mathcal{R}\left( \mathbb{R}^2,\sigma \right)$ on $L^2\left(\mathbb{R}\right)$, $\left[e^{itH} \pi_S\left( a \right) e^{-itH} \right] \in \pi_S\left[\mathcal{R}\left( \mathbb{R}^2,\sigma \right)\right]$ holds for all $a \in \mathcal{R}\left( \mathbb{R}^2,\sigma \right)$ and $\left( H - i\lambda\mathds{1} \right)^{-1} \in \pi_S\left[ \mathcal{R}\left( \mathbb{R}^2,\sigma \right)\right]$, for all $\lambda \in \mathbb{R}\setminus\{0\}$. Results from \cite{01_BG}, \cite{07_HS} and \cite{09_HK} are used for the purpose.
\end{abstract}

\justifying

\section{Introduction}
In 2008, Buchholz and Grundling introduced a novel $\text{C}^{\ast}-$algebra, the resolvent algebra (\cite{01_BG}), to model (bosonic) quantum mechanical systems, aiming to overcome some of the drawbacks afflicting the historically famous Weyl algebra, typically employed for the purpose. In details, given a symplectic vector space $\left(X,\sigma\right)$, the dynamics of the Weyl algebra $\text{CCR}\left( X,\sigma \right)$ defined by symplectic transformations of $X$ correspond to the dynamics induced by quadratic Hamiltonians, hence excluding many physically interesting situations; as a matter of fact, the following result holds.\\

\noindent
\textbf{Proposition (\cite{02_FV}):} Let $\left( H_{\lambda} = H_0 + \lambda V, \, \mathcal{D}_{H_{\lambda}} \right)$ be a self-adjoint operator on $L^2\left( \mathbb{R} \right)$, where $H_0$ is the free Hamiltonian. If $V \in L^{1}\left( \mathbb{R} \right) \cap L^{\infty}\left( \mathbb{R} \right)$, $\forall \lambda \in \mathbb{R} \setminus \left\{ 0\right\}, t \in \mathbb{R}$, the automorphism $\alpha_{t}^{\lambda}\left( \cdot \right) = e^{itH_{\lambda}}\left( \cdot \right) e^{-itH_{\lambda}}$ of $\mathfrak{B}\left(L^2\left( \mathbb{R} \right)\right)$ is not an automorphism of $\text{CCR}\left(\mathbb{R}^{2},\sigma\right)$ unless $V$ is null. $\hfill \blacksquare$ \\

\noindent Moreover, in regular representations, natural observables as bounded functions of the Hamiltonian do not belong to $\text{CCR}\left( X,\sigma \right)$. Resolvent algebras $\mathcal{R}\left(X,\sigma\right)$, on the other hand, have proved to be not affected by similar drawbacks in different interesting situations (\cite{03_B}, \cite{04_B}); a crucial role for this to happen is its non-simple nature (\cite{01_BG}, \cite{05_B}).\\

\noindent This paper takes a step forward in this direction: it shows that, for the case of a single, non-relativistic spinless particle, the resolvent algebra $\mathcal{R}\left( \mathbb{R}^2,\sigma \right)$, on the one hand, can accommodate dynamics induced by fixed located Dirac delta-like potentials, on the other hand, the corresponding Hamiltonians are affiliated to it, respectively meaning that, if $H$ is the whatever point-like Hamiltonian of the case, $e^{itH} \pi_S\left(a\right) e^{-itH} \in \pi_S\left[ \mathcal{R}\left( \mathbb{R}^2,\sigma \right) \right]$, for all $a \in \mathcal{R}\left( \mathbb{R}^2,\sigma\right)$ and $\left( H - i\lambda\mathds{1} \right)^{-1} \in \pi_S\left[ \mathcal{R}\left( \mathbb{R}^2,\sigma \right) \right]$, for all $\lambda \in \mathbb{R}\setminus\{0\}$, where $\pi_S$ is the Schr{\"o}dinger representation of $\mathcal{R}\left( \mathbb{R}^2,\sigma \right)$ on $L^2\left( \mathbb{R} \right)$. \\

\noindent The reason to focus on delta interactions lies in their ubiquity in applications: they represent good candidates to approximate unknown, short-scaled interactions. An extensive mathematical literature has been dedicated to such singular potential: \cite{06_AGH-KH}, for example, provides a very well-known illustrative reference.\\

\noindent The paper is articulated as follows: section 2 recaps the necessary structural results about resolvent algebras, while section 3 first briefly recalls the content of \cite{07_HS}, then proves what above anticipated, with $H$ alternatively given by the symbols

\begin{itemize}
    \item $H = - \frac{d^2}{dx^2} + \alpha \delta\left( x - x_0 \right)$, $\alpha \in \mathbb{R}\setminus\{0\}, \, x_0 \in \mathbb{R}$,
    \item $H = - \frac{d^2}{dx^2} + \sum_{i = 1}^N \, \alpha_i \delta \left( x - x_i \right)$, $\alpha_i \in \mathbb{R} \setminus \{0\}$, $x_i \in \mathbb{R}: \, x_i \neq x_j$, $i,j \in \{1,\ldots,N\}$,
    \item $H = - \frac{d^2}{dx^2} + \sum_{i = 1}^{\infty} \alpha_i \delta \left( x - x_i \right)$, $\left\{ \alpha_i \right\}_i \in l^{1}\left( \mathbb{N} \right)$, $ x_i \in \mathbb{R}: \, x_i \neq x_j, \, i,j \in \mathbb{N}$.
\end{itemize}

\section{The Resolvent Algebra}
\justifying

\begin{definition}\label{Def2.1}
    \justifying
    Let $\left( X, \sigma \right)$ be a symplectic space and $\mathfrak{R} \doteq \left\{ R\left(\lambda, f \right) \mid \lambda \in \mathbb{R} \setminus \{0\}, \, f \in X \right\}$ a set of symbols. Let $L$ be the following list of relations.
    \begin{enumerate}
        \item $R\left( \lambda, 0 \right) = - \frac{i}{\lambda} \mathds{1}$
        \item $R\left( \lambda, f \right)^\ast = R\left( - \lambda, f \right)$
        \item $\nu R\left( \nu \lambda, \nu f \right) = R\left( \lambda, f \right)$
        \item $R\left( \lambda, f \right) - R\left( \mu, f \right) = i \left( \mu - \lambda\right) R\left( \lambda, f \right)R\left( \mu, f \right)$
        \item $\left[ R\left( \lambda, f \right), R\left( \mu, g \right)\right] = i \sigma\left( f,g \right) R\left( \lambda, f \right) R\left( \mu, g \right)^2 R\left( \lambda, f \right)$
        \item $R\left( \lambda, f \right) R\left( \mu, g \right) = R\left( \lambda + \mu, f+g \right) \left[ R\left( \lambda, f \right) + R\left( \mu, g \right) +  i\sigma\left( f, g \right) R\left( \lambda, f \right)^2 R\left( \mu, g \right) \right]$
    \end{enumerate}
    with $\lambda, \mu, \nu \in \mathbb{R} \setminus \{0\}$, $f,g \in X$ and $\lambda + \mu \neq 0$ at $6.$. Denoted by $\tilde{\mathcal{R}}_{0}$ the unital $\ast-$algebra freely generated by $\mathfrak{R}$ and the two-sided ideal $\mathcal{I}_{L}$ generated by the given relations, the \textbf{unital pre-resolvent $\ast-$algebra} $\mathcal{R}_0 \doteq \tilde{\mathcal{R}}_{0} / \mathcal{I}_L$ is defined. $\hfill \square$
\end{definition}

\begin{remark}
    \justifying Turning $\mathcal{R}_0$ into a $C^{\ast}-$algebra amounts in defining a $C^{\ast}-$norm on it. The following result is preparatory. $\hfill \square$
\end{remark}

\begin{theorem}\label{Prop. 2.1}
    \justifying 
    Let $\left( X, \sigma \right)$ be a symplectic space and let $\mathcal{R}_0$ be as in definition \ref{Def2.1}.
    \begin{enumerate}
        \item If $\mathcal{H}$ is a Hilbert space and $\pi_0: \, \mathcal{R}_0 \longrightarrow \mathcal{B}\left( \mathcal{H} \right)$ is a (bounded) $\ast-$representation of $\mathcal{R}_0$, then $\norm{\pi_0 \left[R\left(\lambda, f \right)\right]} \leq \abs{\lambda}^{-1}$. In other words, for all $a \in \mathcal{R}_0$, there exists $c_a \geq 0$, only depending on $a$, such that $\norm{\pi\left(a\right)} \leq c_a$, for all the bounded Hilbert space representations $\pi$ of $\mathcal{R}_0$.
        \item Let $\omega$ be a positive functional of $\mathcal{R}_0$, i.e. $\omega: \mathcal{R}_0 \longrightarrow \mathbb{C}$ such that $\omega\left(a^\ast a\right) \geq 0$ for all $a \in \mathcal{R}_0$. Then, the GNS-construction yields a cyclic $\ast-$representation of $\mathcal{R}_0$, denoted by $\pi_{\omega}$, consisting of bounded Hilbert space operators.
    \end{enumerate}
    $\hfill \blacksquare$
\end{theorem}

\begin{remark}
    \justifying
    The importance of the foregoing result lies in the following fact: a $\ast-$algebra can have unbounded representations, but, as long as the GNS-construction is concerned, the corresponding representation is bounded. $\hfill \square$
\end{remark}

\begin{definition}
    \justifying
    Let $\mathcal{R}_0$ over $\left( X,\sigma \right)$ be as in Definition \ref{Def2.1}. Its universal representation $\pi_u: \mathcal{R}_0 \longrightarrow \mathcal{B}\left(\mathcal{H}\right)$ is given by
    \begin{equation*}
        \pi_u\left( a \right) \doteq \bigoplus \left\{ \pi_\omega(a) \mid \omega \in \mathfrak{S} \right\} \quad \text{and} \quad \norm{a}_u \doteq \norm{\pi_u\left( a \right)} \equiv \underset{\omega \in \mathfrak{S}}{\sup} \, \norm{\pi_{\omega} \left( a \right)}
    \end{equation*}
    \noindent denotes the corresponding enveloping $C^{\ast}-$norm, where $\mathfrak{S}$ is the set of all positive and normalized\footnote{i.e. $\omega{\left(\mathds{1}\right)} = 1$} (algebraic) states over $\mathcal{R}_0$. The completion of $\mathcal{R}_0/ \ker \, \norm{\cdot}_u$ with respect to $\norm{\cdot}_u$ is denoted by $\mathcal{R}\left(X,\sigma \right)$ and defines the \textbf{resolvent algebra} over $\left(X,\sigma \right)$. $\hfill \square$
    
\end{definition}

\begin{remark}
    \justifying
    Because of Proposition \ref{Prop. 2.1}(1), $\norm{\cdot}_u$ is well-defined, because of $\norm{\pi_\omega \left( a \right)} \leq c_a$, for all $\omega \in \mathfrak{S}$, hence
    \begin{equation*}
        \underset{\omega \in \mathfrak{S}}{\sup} \, \norm{\pi_\omega \left( a \right)} \leq c_a < \infty.
    \end{equation*}
    $\hfill \square$
\end{remark}

\begin{remark}
    \justifying
    It is reported that, given $f \in X$, the map
    \begin{equation*}
        R\left(\cdot, f \right): \, \lambda \in \mathbb{R} \setminus \{0\} \longmapsto R\left(\lambda, f \right) \in \mathcal{R}\left( X,\sigma \right)
    \end{equation*}
    is analytic. Such a circumstance allows for the possibility of having complex-valued $\lambda$ arguments for $R\left( \lambda, f \right)$ as long as these are not purely imaginary. Differently said, $\mathcal{R}\left(X,\sigma \right)$ contains $R\left( z, f \right)$, with $z \in \mathbb{C} \setminus i\mathbb{R}$ too.
    $\hfill \square$
\end{remark}

\begin{definition}
    \justifying
    Let $\mathcal{H}$ be a Hilbert space and $\left( X,\sigma \right)$ be a symplectic vector space. Given $S \subseteq X$, a representation $\pi: \, \mathcal{R}\left( X, \sigma \right) \longrightarrow \mathcal{B}\left( \mathcal{H} \right)$ is said \textbf{regular on} $S$ if and only if $\ker\left\{ \pi \left[ R \left(1, f \right) \right]\right\} = \left\{ 0 \right\}$ for all $f \in S$. Consequently, a state $\omega$ of $\mathcal{R}\left( X, \sigma \right)$ is said  \textbf{regular on} $S$ if and only if its GNS-representation is regular on $S$. A representation (state) is simply said \textbf{regular} if it is regular on $X$. The set of all regular representations of $\mathcal{R}\left( X, \sigma \right)$ on $\mathcal{H}$ is denoted by $\text{reg}\left( \mathcal{R}\left(X, \sigma \right), \mathcal{H} \right)$, while the set of all regular states by $\mathfrak{S}_r\left(\mathcal{R} \left(X, \sigma \right)\right)$.
    $\hfill \square$
\end{definition}

\begin{remark}
    \justifying
    The importance of the notion of regular representation stems from the fact that, given the Hilbert space $\mathcal{H}$ together with the regular representation $\pi: \, \mathcal{R}\left( X, \sigma \right) \longrightarrow \mathcal{B} \left(\mathcal{H}\right)$, for all $f \in X$, one can define the linear operator $\left(\phi_{\pi} \left( f \right), \, \mathcal{D}_{\phi_{\pi} \left( f \right)} \right)$, where
    \begin{align*}
        \mathcal{D}_{\phi_{\pi}} \left( f \right) & = \text{Ran}\left[ R \left(1, f \right) \right] \\
        \phi_{\pi} \left( f \right) & = R\left( 1, f\right)^{-1} + i\mathds{1}.
    \end{align*}
    Such an operator is proved to be self-adjoint and its definition is independent on the choice of $\lambda \equiv 1$. $\hfill \square$
\end{remark}

\begin{theorem}\label{Prop. 2.2}
    \justifying
    Let $\left(X,\sigma\right)$ be a symplectic vector space and $\mathcal{R}\left( X, \sigma \right)$ be the corresponding resolvent algebra. Every regular representation of $\mathcal{R}\left( X, \sigma \right)$ is faithful. $\hfill \blacksquare$
\end{theorem}

\begin{theorem}\label{Prop. 2.3}
    \justifying
    Let $\mathcal{H}$ be a Hilbert space and $\left(X, \sigma \right)$ be a symplectic vector space. Given a regular representation $\pi \in \text{reg} \left( \mathcal{R}\left(X, \sigma \right), \mathcal{H} \right)$, the mapping
    \begin{equation*}
        \tilde{\pi}: \, \delta_f \in \text{CCR} \left( X, \sigma \right) \longmapsto \tilde{\pi}\left( \delta_f \right) \doteq e^{i\phi_{\pi}\left( f \right)} \in \mathcal{B}\left( \mathcal{H} \right)
    \end{equation*}
    results in a regular representation of the Weyl algebra over $\left(X, \sigma \right)$; particularly,
    \begin{equation*}
        \pi \in \text{reg} \left( \mathcal{R}\left(X, \sigma \right), \mathcal{H} \right) \longmapsto \tilde{\pi} \in \text{reg}\left( \text{CCR} \left(X, \sigma \right), \mathcal{H} \right)
    \end{equation*}
    is a bijection preserving irreducibility and direct sums; its inverse is defined by
    \begin{equation*}
        \pi \left[ R\left(\lambda, f \right) \right] \doteq -i \int_0^{\eta \infty} \, e^{-\lambda t} \Tilde{\pi}\left(\delta_{-tf}\right) dt, \quad \eta \equiv sign(\lambda)
    \end{equation*}
    \noindent where the integral is understood in the strong operator topology. $\hfill \blacksquare$
\end{theorem}

\begin{remark}
    \justifying
    The foregoing result allows to state that, as long as $\left( X, \sigma \right)$ is finite-dimensional, the Stone-von Neumann theorem holds for the resolvent algebra too, i.e. $\mathcal{R}\left( X, \sigma \right)$ admits a unique (up to unitary equivalence) irreducible, regular representation, the Schr{\"o}dinger representation. $\hfill \square$
\end{remark}

\begin{theorem}\label{Prop. 2.4}
    \justifying
    Let $\left( X, \sigma \right)$ be a finite-dimensional symplectic space and $\pi_0: \, \mathcal{R}\left( X, \sigma \right) \longrightarrow \mathcal{B}\left( \mathcal{H}_0 \right)$ be an irreducible regular representation of $\mathcal{R}\left( X, \sigma \right)$. Then, there exists a unique closed two-sided ideal $\mathcal{K}$ of $\mathcal{R}\left( X, \sigma \right)$ isomorphic to the algebra of compact operators $\mathcal{B}_{\infty}\left( \mathcal{H}_0 \right)$, such that $\pi_0\left(\mathcal{K}\right) = \mathcal{B}_{\infty}\left(\mathcal{H}_0\right) \subset \mathcal{B}\left( \mathcal{H}_0 \right)$.
    $\hfill \blacksquare$
\end{theorem}

\subsection{Dynamics on $\mathcal{R}\left(X, \sigma \right)$ }
\justifying

\begin{definition}\label{Dynamics}
    \justifying
    Let $\left( X, \sigma \right) \equiv \left( \mathbb{R}^{2N}, \sigma \right), \, N \in \mathbb{N}$ be the standard symplectic vector space and $\pi_S: \, \mathcal{R}\left( X, \sigma \right) \longrightarrow \mathcal{B}\left( L^2\left(\mathbb{R}^N\right) \right)$ be the Schr{\"o}dinger representation of $\mathcal{R}\left( \mathbb{R}^{2N},\sigma \right)$. A $\left( H, \mathcal{D}_{H} \right)$ self-adjoint Hamiltonian on $L^2\left(\mathbb{R}^N\right)$ \textbf{induces a dynamics on} $\mathcal{R}\left(X,\sigma \right)$ if and only if
    \begin{equation*}
        e^{itH} \left\{\pi_S \left[ \mathcal{R} \left( X, \sigma \right) \right] \right\} e^{-itH} \subseteq \pi_S \left[ \mathcal{R} \left( X, \sigma \right) \right], \quad \forall t \in \mathbb{R}.
    \end{equation*}
    $\hfill \square$
\end{definition}

\begin{remark}
    \justifying
    Such a definition is well-posed, because of the regularity, hence the injectivity, of the Schr{\"o}dinger representation, together with the fact that, because of the Stone-von Neumann theorem, such a representation is unique up to unitary equivalence. $\hfill \square$
\end{remark}

\begin{theorem}\label{Prop. 2.5}
    \justifying
    Let $\mathcal{R} \left( \mathbb{R}^2, \sigma \right)$ be and $\pi_S: \mathcal{R} \left( \mathbb{R}^2, \sigma \right) \longrightarrow \mathcal{B}\left( L^2\left(\mathbb{R}\right) \right)$ be the corresponding Schr{\"o}dinger representation. Given $V \in C_0\left( \mathbb{R} \right)$, the self-adjoint Hamiltonian $\left(H \equiv H_0 + V, \mathcal{D}_{H} \right)$ induces a dynamics on $\mathcal{R}\left( \mathbb{R}^2, \sigma \right)$. Moreover $R_{H}\left( \lambda \right) = \left( H - i\lambda \mathds{1} \right)^{-1} \in \pi_S\left[\mathcal{R}\left(\mathbb{R}^2, \sigma\right)\right]$ for all $\lambda \in \mathbb{R} \setminus \{0\}$, i.e. $H$ is affiliated to $\mathcal{R}\left(\mathbb{R}^2, \sigma\right)$. $\hfill \blacksquare$
\end{theorem}

\section{Point Interactions}
\justifying

Proposition \ref{Prop. 2.5} guarantees that, for one-dimensional quantum mechanical systems, Schr{\"o}dinger Hamiltonians with $C_0\left( \mathbb{R} \right)$ potentials do induce dynamics on $\mathcal{R}\left( \mathbb{R}^2, \sigma \right)$. This section, on the other hand, addresses the stability problem of $\mathcal{R}\left( \mathbb{R}^2, \sigma \right)$ under the action of symbolic Hamiltonians\footnote{In this section, $H_0 = -\frac{d^2}{dx^2}$ is assumed.} as
    \begin{equation}\label{Hamiltonian Fixed Centers}
        H = -\frac{d^2}{dx^2} + \sum_{i=1}^{N} \, \alpha_i \delta\left(x - x_i\right),
    \end{equation}
    \noindent with $N \in \mathbb{N} \cup \{\infty\}, \; x_i \in \mathbb{R}: \, x_i \neq x_j, \, \forall \, i,j, \; \alpha_i \in \mathbb{R}\setminus \{0\}, \, \forall i$, clearly not of Schr{\"o}dinger type. Definition \ref{Dynamics} requires showing
    \begin{equation}\label{thesis}
        e^{itH}\pi_S\left( a \right)e^{-itH} \in \pi_S \left[ \mathcal{R}\left( \mathbb{R}^2, \sigma \right)\right], \quad \forall a \in \mathcal{R}\left( \mathbb{R}^2,\sigma \right), \; \forall t \in \mathbb{R},
    \end{equation}
    \noindent hence the first issue to be dealt with is the explicit construction of $e^{-itH}$, $t \in \mathbb{R}$ given $H$ as in (\ref{Hamiltonian Fixed Centers}). \cite{07_HS} is extremely useful for the announced purpose; concretely, by observing that
    \begin{align*}
        \Gamma_{V}\left( t \right) & = \mathds{1} + \sum_{n \in \mathbb{N}} \, \left( -i \right)^{n} \int_0^{t} \, dt_n \cdots \int_0^{t_2} \, dt_1 \, V\left( t_1 \right) \cdots V\left( t_n \right) \\
        & = \mathds{1} + \sum_{n \in \mathbb{N}} \, \left( -i \right)^{n} \int_0^{t} \, dt_n \cdots \left\{ \int_0^{t_3} \, dt_2 \left[ \int_0^{t_2} \, dt_1 V\left( t_1 \right) \right] V\left( t_2 \right) \cdots \right\} V\left( t_n \right) \\
        & = \mathds{1} + \sum_{n \in \mathbb{N}} \, \left( - i \right)^{n} \int_0^{t} \, dt_n \; \Gamma_{V, \; \left(n-1\right)}\left( t_n \right) V\left( t_n \right) \equiv \mathds{1} + \sum_{n \in \mathbb{N}} \, \left(-i\right)^{n} \, \Gamma_{V, \; \left(n\right)}\left( t \right),
    \end{align*}
    for all $t \in \mathbb{R}$, i.e. by considering
    \begin{equation}\label{Interaction Picture 03}
        \Gamma_{V, \; \left(n\right)} \left(t\right) = \int_0^{t} \, dt_n \; \Gamma_{V, \; \left(n-1\right)}\left( t_n \right) V\left( t_n \right), \quad t \in \mathbb{R}, \, n \in \mathbb{N}
    \end{equation}
    \noindent with $\Gamma_{V, \; \left(0\right)}\left(t\right) = \mathds{1}$ for all $t \in \mathbb{R}$, $\Gamma_{V}\left( t \right)$ can be explicitly built out of the Fourier domain counterparts of $\Gamma_{V, \; \left(n\right)} \left(t\right)$, in turn defined by the integral kernels 
    \begin{align}
        K_{t, \, (1)}\left( p,q \right) & = \left[ \mathfrak{F} \Gamma_{V, \; \left(1\right)}\left( t \right) \mathfrak{F}^{-1} \right]\left(p,q \right) =  \left[\frac{e^{it\left(p^{2}-q^{2}\right)} - 1}{i\left(p^{2}-q^{2}\right)}\right] \frac{\tilde{V}\left(p-q\right)}{\sqrt{2\pi}} \label{Kernel 1}\\
        K_{t, \, (n)}\left(p,q\right) & = \left[ \mathfrak{F} \Gamma_{V, \; \left(n\right)}\left( t \right) \mathfrak{F}^{-1} \right]\left(p,q \right) =  \int_{0}^{t} dt_n \int_{\mathbb{R}} dz_{n-1} K_{t_n, \, (n-1)}\left(p, z_{n-1}\right) e^{it_{n} \left( z_{n-1}^{2} - q^{2} \right) } \frac{\tilde{V}\left( z_{n-1} - q \right)}{\sqrt{2\pi}}, \label{Kernel n}
    \end{align}
    \noindent $t,p,q \in \mathbb{R}, \, n \in \mathbb{N}: n \geq 2$, where $\mathfrak{F}$ is the Fourier-Plancherel operator. $V$ is then allowed to be a distribution on $\mathbb{R}$ whose Fourier transform $\Tilde{V}$ is a $L^{\infty}$ function such that $\overline{\Tilde{V}\left( p \right)} = \Tilde{V}\left( - p \right), \; p \in \mathbb{R}$. Further, \cite{07_HS} ensures that each $K_{t, \, (n)}$ is a bounded operator on $L^2\left( \mathbb{R} \right)$ for all $t \in \mathbb{R}$ and that the Dyson series $\sum_{n \in \mathbb{N}} \, K_{t, \, \left( n \right)}$ converges in the uniform norm topology; set, then, $K\left(t\right) = \mathds{1} + \sum_{n \in \mathbb{N}} \, K_{t, \, \left( n \right)}$, 
    \begin{equation*}
        U\left(t\right) = e^{-itH_0}\left[\mathfrak{F}^{-1} K\left(t\right) \mathfrak{F}\right], \quad t \in \mathbb{R}
    \end{equation*}
    \noindent gives the rigorous unitary time evolution operator of a system governed by the symbolic Hamiltonian $H = H_0 + V$. $\hfill \square$


\subsection{One Fixed-Center Point Interaction}

One spinless particle undergoing a unique point interaction placed in a fixed location of the real line is considered. Its formal Hamiltonian is
\begin{equation}
    H =  -\frac{d^2}{dx^2} + \alpha \delta(x-x_0), \quad \, x_0 \in \mathbb{R}, \label{Single Point Hamiltonian}
\end{equation}

\noindent where $\alpha \in \mathbb{R}\setminus \{0\}$ is the coupling constant and $x_0 \in \mathbb{R}$ is the $\delta-$location. Given $V = \alpha \delta\left(\cdot - x_0\right)$, one has
\begin{itemize}
    \item $\tilde{V}\left( p \right) = \frac{\alpha}{\left(\sqrt{2\pi}\right)} e^{-ipx_0}, \; p \in \mathbb{R} \implies \tilde{V} \in L^{\infty}\left( \mathbb{R} \right)$,
    \item $\overline{\Tilde{V}\left( p \right)} = \overline{\frac{\alpha}{\left(\sqrt{2\pi}\right)} e^{-ipx_0}} \equiv \frac{\alpha}{\left(\sqrt{2\pi}\right)} e^{-i \left( - p\right) x_0} \equiv \Tilde{V}\left( - p \right), \quad p \in \mathbb{R}$.
\end{itemize}
Consequently,
\begin{align}
    K_{t, \, (1)}^{\left(\alpha\right)} \left( p,q \right) & = \frac{\alpha}{2\pi} \left[\frac{e^{it\left(p^{2}-q^{2}\right)} - 1}{i\left(p^{2}-q^{2}\right)}\right] e^{-i(p-q)x_0} \label{eq: Kernel 01 Delta} \\
    K_{t, \, (n)}^{\left(\alpha\right)} \left( p,q \right) & = \frac{\alpha}{2\pi} \int_{0}^{t}dt_n \int_{\mathbb{R}} dz_{n-1} \; K_{t_n, \, (n-1)}^{\left(\alpha\right)}\left(p - z_{n-1}\right) e^{it_{n}\left( z_{n-1}^{2} - q^2\right)} e^{-i(z_{n-1} - q)x_0}, \quad n \in \mathbb{N} \label{eq: Kernel 02 Delta}
\end{align}
\noindent allow to build $\Gamma_{V}(t) \equiv \Gamma_{\alpha}(t)$ as described, hence the unitary time evolution operator $U_{\alpha}\left(t\right) = e^{-itH_0} \Gamma_{\alpha}(t)$ corresponding to (\ref{Single Point Hamiltonian}), for all $t \in \mathbb{R}$. $\hfill \square$

\begin{remark}\label{Clue Remark 1P}
    \justifying
    Endowed with the unitary time evolution operator, the resolvent algebra $\mathcal{R}\left( \mathbb{R}^2, \sigma \right)$ stability remains to be proved; the following strategy is adopted: given a non-negative smooth function of compact support $W$\footnote{It does not harm generality assuming $\int_{\mathbb{R}} W = 1$}, introducing $W_{\epsilon}$ as
    \begin{equation*}
        W_{\epsilon}: \, x \in \mathbb{R} \longmapsto W_{\epsilon}\left(x\right) \doteq \frac{1}{\epsilon}W\left( \frac{x}{\epsilon} \right) \in \mathbb{R}, \quad \epsilon > 0,
    \end{equation*}
    along with the Schr{\"o}dinger Hamiltonian $\left(H_{\epsilon} = H_0 + \alpha W_{\epsilon}, \, \mathcal{D}_{H_0} \right)$, Proposition \ref{Prop. 2.5} allows to claim that $\exp\left(-itH_{\epsilon}\right) \in \pi_S\left[\mathcal{R}\left(\mathbb{R}^2,\sigma\right)\right]$ for all $t \in \mathbb{R}$. Therefore, should
    \begin{equation}\label{purpose1}
        \norm{U_\alpha \left( t \right) - e^{-itH_{\epsilon}}}_{\mathcal{B}\left(L^2\left(\mathbb{R}\right)\right)} \underset{\epsilon \downarrow 0}{\longrightarrow} 0
    \end{equation}
    hold, the stability of $\mathcal{R}\left(\mathbb{R}^2,\sigma \right)$ would easily follow; infact, for all $a \in \mathcal{R}\left( \mathbb{R}^2,\sigma \right), \, t \in \mathbb{R}$,
    \begin{align*}
        & \norm{U_\alpha \left( t \right)^{\ast} \pi_S\left( a \right) U_\alpha \left( t \right) - e^{itH_{\epsilon}} \pi_S \left(a \right) e^{-itH_{\epsilon}} } = \\
        = & \norm{ U_\alpha \left( t \right)^{\ast} \pi_S\left( a \right) U_\alpha \left( t \right) - e^{itH_{\epsilon}} \pi_S \left(a \right) U_\alpha \left( t \right) + e^{itH_{\epsilon}} \pi_S\left( a \right) U_\alpha \left( t \right) - e^{itH_{\epsilon}} \pi_S \left(a \right) e^{-itH_{\epsilon}}} \leq \\
        \leq & \norm{U_\alpha \left( t \right)^{\ast} - e^{itH_{\epsilon}}} \norm{\pi_S\left( a \right)} + \norm{\pi_S\left( a \right)} \norm{U_\alpha \left( t \right) - e^{-itH_{\epsilon}}} \underset{\epsilon \downarrow 0}{\longrightarrow} 0.
    \end{align*}
    $\hfill \square$
\end{remark}

\begin{remark}\label{Auxiliary for One Point Case Result}
    \justifying
    \begin{enumerate}
        \item $\left\{W_{\epsilon}\right\}_{\epsilon > 0}$ converges to $\delta$ in $\mathscr{D}^{\, \prime}(\mathbb{R})$ as $\epsilon \rightarrow 0$; infact, given a whatever compactly supported real smooth function $f$ on $\mathbb{R}$,
        \begin{align*}
            \underset{\epsilon \downarrow0}{\lim} \, \int_{\mathbb{R}} \left[\frac{1}{\epsilon} W\left(\frac{x-x_0}{\epsilon}\right)\right] f\left(x\right) dx & = \int_{\mathbb{R}} W\left(x\right) \left[\underset{\epsilon\downarrow0}{\lim}f\left(\epsilon x+x_0\right)\right] dx = \left(\int_{\mathbb{R}}W\right) f\left(x_0\right) \equiv \\ & \equiv \int_{\mathbb{R}} \left[ \delta \left(x-x_0\right)\right] f\left(x\right)dx,
        \end{align*}
        by using the Lebesgue dominated convergence theorem and $\int_{\mathbb{R}} W = 1$.
        \item $\left\{ \mathcal{F}\left[ W_{\epsilon}\left(\cdot - x_0\right)\right] \right\}_{\epsilon}$\footnote{$\mathcal{F}$ denotes the $L^1$-Fourier transform operator.} is point-wise convergent to $\left[\left(\sqrt{2\pi}\right)^{-1} e^{-i\left( \cdot \right)x_0} \right]$ as $\epsilon \downarrow 0$ for all $x_0 \in \mathbb{R}$ and there exists $M \in \mathbb{R}^+$ such that $\abs{ \tilde{W}_{\epsilon} \left(p\right)} \leq M$, for all $\epsilon$ and $p$: straightforwardly,
        \begin{equation*}
            \underset{\epsilon\downarrow0}{\lim} \; \mathcal{F}\left[ W_{\epsilon}\left(\cdot - x_0\right)\right] \left( p \right) \equiv \underset{\epsilon\downarrow0}{\lim} \; \tilde{W}_{\epsilon}\left(p\right) =  \underset{\epsilon\downarrow0}{\lim} \; \frac{e^{-ipx_0}}{\sqrt{2\pi}} \int_{\mathbb{R}}W\left(x\right) e^{-i\left(\epsilon p\right)x} dx = \frac{e^{-ipx_0}}{\sqrt{2\pi}},\quad \forall p \in \mathbb{R},
        \end{equation*}
        by the Lebesgue dominated convergence theorem. Clearly
        \begin{equation*}
            \abs{\Tilde{W_\epsilon}(p)} \leq \frac{1}{\sqrt{2\pi}} \equiv M.
        \end{equation*} 
    \end{enumerate}
    $\hfill \square$
\end{remark}

\begin{theorem}\label{One Fixed-Center Result}
Let $\alpha \in \mathbb{R} \setminus \{0\}$ be and $W \in C^{\infty}_c \left( \mathbb{R}\right)$ as in Remark \ref{Clue Remark 1P}. For all $t \in \mathbb{R}$,
\begin{equation*}
    \norm{U_\alpha \left( t \right) - e^{-itH_\epsilon}}_{\mathcal{B}\left(L^2\left(\mathbb{R}\right)\right)} \underset{\epsilon\downarrow0}{\longrightarrow}0.
\end{equation*}

\noindent holds.

\end{theorem}

\begin{proof}
Set
\begin{align}
    K_{t, \, (1)}^{\left(\epsilon \right)} \left( p,q \right) & = \frac{\alpha}{\sqrt{2\pi}} \left[\frac{e^{it\left(p^{2}-q^{2}\right)} - 1}{i\left(p^{2}-q^{2}\right)}\right] \Tilde{W}_{\epsilon}\left(p-q\right) \label{Kernel Approssimante 1 un Centro} \\
    K_{t, \, (n)}^{\left(\epsilon\right)} \left( p,q \right) & = \frac{\alpha}{\sqrt{2\pi}} \int_{0}^{t}dt_{n} \int_{\mathbb{R}}dz_{n-1}  K_{t_n, \, (n-1)}^{\left(\epsilon \right)}\left(p,z_{n-1}\right) e^{it_{n}\left( z_{n-1}^{2} - q^{2} \right)} \Tilde{W}_{\epsilon}\left( z_{n-1} - q\right),\quad n \in \mathbb{N} \label{Kernel Approssimante n un Centro}
\end{align}
\noindent it results
\begin{align*}
\norm{U_\alpha \left( t \right) - e^{-itH_\epsilon}} & = \\ = \norm{e^{-itH_0} \left[\mathds{1} + \sum_{n \in \mathbb{N}} \, i^n \mathfrak{F}^{-1} K_{t, \, \left( n \right)}^{(\alpha)} \mathfrak{F} \right] - e^{-itH_0} \left[\mathds{1} + \sum_{n \in \mathbb{N}} \, i^n \mathfrak{F}^{-1} K_{t, \, \left( n \right)}^{(\epsilon)} \mathfrak{F} \right]} & \leq \sum_{n\in\mathbb{N}} \norm{K_{t, \, \left( n \right)}^{\left(\alpha\right)} - K_{t, \, \left( n \right)}^{\left(\epsilon\right)}},
\end{align*}

\noindent meaning that proving the claim amounts in showing
\begin{equation*}
    \norm{ K_{t, \, \left(n\right)}^{\left( \alpha \right)} - K_{t, \, \left(n\right)}^{\left( \epsilon \right)}} \underset {\epsilon \downarrow 0}{\longrightarrow}0, \quad \forall n \in \mathbb{N}, \; \forall t \in \mathbb{R}.
\end{equation*}

\noindent \cite{07_HS} thm. 3.4 allows for
\begin{align*}
    \norm{K_{t, \, \left( n \right)}^{\left(\alpha\right)} - K_{t, \, \left( n \right)}^{\left(\epsilon \right)}} & \leq \left\{ \left(\underset{p\in\mathbb{R}}{\sup} \int_{\mathbb{R}} \abs{K_{t, \, \left( n \right)}^{\left(\alpha\right)}\left(p,q\right) - K_{t, \, \left( n \right)}^{\left(\epsilon\right)}\left(p,q\right)} dq \right)\right.\\ & \left.\left(\underset{p\in\mathbb{R}}{\sup} \int_{\mathbb{R}} \abs{K_{t, \, \left( n \right)}^{\left(\alpha\right)\ast}\left(p,q\right) - K_{t, \, \left( n \right)}^{\left(\epsilon\right)\ast}\left(p,q\right)} dq \right) \right\}^{\frac{1}{2}} < \infty, \quad n \in \mathbb{N}
\end{align*}

\noindent hence the \emph{induction principle} is going to be used.

\paragraph{$\boxed{k=1}$}
\begin{equation*}
    \underset{\epsilon\downarrow0}{\lim} \int_{\mathbb{R}} \abs{K_{t, \, \left( 1 \right)}^{\left(\alpha\right)}\left(p,q\right) - K_{t, \, \left( 1 \right)}^{\left(\epsilon\right)}\left(p,q\right)} dq = \frac{\abs{\alpha}}{\sqrt{2\pi}} \; \underset{\epsilon\downarrow0}{\lim} \int_{\mathbb{R}} \abs{\frac{e^{it\left(p^{2}-q^{2}\right)}-1}{p^{2}-q^{2}}} \abs{\frac{e^{-i(p-q)x_0}}{\sqrt{2\pi}} - \tilde{W_{\epsilon}}\left(p-q\right)} dq
\end{equation*}

\noindent is intended to be studied. By observing that
\begin{equation*}
    \frac{\abs{\alpha}}{\sqrt{2\pi}} \int_{\mathbb{R}} \abs{\frac{e^{it\left(p^{2}-q^{2}\right)}-1}{p^{2}-q^{2}}} \abs{\frac{e^{-i(p-q)x_0}}{\sqrt{2\pi}} - \tilde{W_{\epsilon}}\left(p-q\right)} dq \leq \sqrt{\frac{2}{\pi}} \abs{\alpha}M \int_{\mathbb{R}} \abs{\frac{e^{it\left(p^{2}-q^{2}\right)}-1}{p^{2}-q^{2}}} dq < \infty
\end{equation*}

\noindent because of \cite{07_HS} thm. 2.3, remark \ref{Auxiliary for One Point Case Result} and the dominated convergence theorem allow for
\begin{equation*}
    \int_{\mathbb{R}} \abs{\frac{e^{it\left(p^{2}-q^{2}\right)}-1}{p^{2}-q^{2}}} \left\{\underset{\epsilon\downarrow0}{\lim} \abs{\frac{e^{-i(p-q)x_0}}{\sqrt{2\pi}} - \tilde{W_{\epsilon}}\left(p-q\right)}\right\} dq = 0
\end{equation*}

\noindent Since $K_{t, \; \left(1\right)} = K_{t, \; \left(1\right)}^{\ast}$,
\begin{equation*}
    \underset{\epsilon\downarrow0}{\lim} \int_{\mathbb{R}} \left| K_{t, \, \left( 1 \right)}^{\ast\left(\alpha\right)}\left(p,q\right) - K_{t, \, \left( 1 \right)}^{\ast\left(\epsilon\right)}\left(p,q\right) \right| dq = 0
\end{equation*}

\noindent holds all the same.

\paragraph{$\boxed{k=n}$}

It is assumed the statement holds for $k \leq n-1$.
{\small{
\begin{align}
    & K_{t, \, \left( n \right)}^{\left(\alpha\right)}\left(p,q\right) - K_{t, \, \left( n \right)}^{\left(\epsilon\right)}\left(p,q\right) = \\ & = \frac{\alpha}{\sqrt{2\pi}} \int_{0}^{t}dt_n \int_{\mathbb{R}}dz_{n-1} e^{it_{n}\left( z_{n-1}^{2} - q^{2} \right)} K_{t_n, \, \left( n-1 \right)}^{\left(\alpha\right)}\left(p, z_{n-1}\right) \left[ \frac{e^{-i(z_{n-1} - q)x_0}}{\sqrt{2\pi}} - \tilde{W_{\epsilon}}\left(z_{n-1} - q\right) \right] + \label{eq: Integrale Delta 01}\\& + \frac{\alpha}{\sqrt{2\pi}} \int_{0}^{t} dt_n \int_{\mathbb{R}} dz_{n-1} e^{it_{n}\left( z_{n-1}^{2} - q^{2}\right)} \left[ K_{t_n, \, \left( n-1 \right) }^{\left(\alpha\right)}\left(p, z_{n-1}\right) - K_{t_n, \, \left(n-1\right)}^{\left(\epsilon\right)}\left(p,z_{n-1}\right)\right] \tilde{W_{\epsilon}}\left(z_{n-1} - q\right).\label{eq: Integrale Delta 02}
\end{align}}}

\noindent A priori, the foregoing integrals are \emph{double} integrals; to use them as \emph{iterated}, Fubini theorem hypotheses have to be ascertained.
\begin{align*}
    \int_{0}^{t} \int_{\mathbb{R}} \abs{K_{t_n, \, \left( n-1 \right)}^{\left(\alpha\right)}\left(p,z_{n-1}\right)} \abs{\frac{e^{-i(z_{n-1} - q)x_0}}{\sqrt{2\pi}} - \tilde{W_{\epsilon}}\left(z_{n-1} - q\right)} dz_{n-1}dt_{n} & \leq \\ \leq 2M \int_{0}^{t} \int_{\mathbb{R}} \abs{K_{t_n, \, \left( n-1 \right)}^{\left(\alpha\right)}\left(p,z_{n-1}\right)} dz_{n-1} dt_{n} < \infty
\end{align*}

\noindent by being $K_{t, \, \left( n \right)}^{\left(\alpha\right)} \in \mathfrak{B}\left(L^{2}\left(\mathbb{R}\right)\right), \, \forall  n \in \mathbb{N}, \, \forall t \in \mathbb{R}$. Fubini also holds for both (\ref{eq: Integrale Delta 02}) and the adjoint case, as can be readily verified. Then, to compute
\begin{equation*}
    \underset{\epsilon\downarrow0}{\lim} \, \int_{\mathbb{R}} \abs{K_{t, \, \left( n \right)}^{\left(\alpha\right)}\left(p,q\right) - K_{t, \, \left( n \right)}^{\left(\epsilon\right)}\left(p,q\right)} dq,
\end{equation*}

\noindent the dominated convergence theorem hypotheses need to be checked out. Therefore
\begin{align*}
    \abs{\int_{0}^{t} \int_{\mathbb{R}} K_{t_n, \, \left( n-1 \right)}^{\left(\alpha\right)}\left(p, z_{n-1}\right) \left[\frac{ e^{-i(z_{n-1} - q)x_0}}{\sqrt{2\pi}} - \tilde{W_{\epsilon}}\left(z_{n-1} - q\right)\right] e^{it_{n}\left( z_{n-1}^{2} - q^{2} \right)} dz_{n-1}dt_{n}} & \leq \\ \leq \abs{\int_{\mathbb{R}} \left[\frac{ e^{-i(z_{n-1} - q)x_0}}{\sqrt{2\pi}} - \tilde{W_{\epsilon}}\left(z_{n-1} - q\right)\right] \int_{0}^{t}  K_{t_n, \, \left( n-1 \right)}^{\left(\alpha\right)}\left(p,z_{n-1}\right) e^{it_{n}\left(z_{n-1}^{2} - q^2\right)} dt_{n}dz_{n-1}} & \leq \\ \leq 2M  \int_{\mathbb{R}} \abs{\int_{0}^{t}  K_{t_n, \, \left( n-1 \right)}^{\left(\alpha\right)}\left(p,z_{n-1}\right) e^{it_{n}\left(z_{n-1}^{2} - q^2\right)} dt_{n}} dz_{n-1} & \equiv \\ \equiv \left(2M\right) \tilde{K}_{t, \, \left( n \right) }^{\left(\alpha\right)} \left(p,q\right),
\end{align*}

\noindent i.e.

\begin{align*}
    \int_{\mathbb{R}} \, \abs{\int_{0}^{t} \int_{\mathbb{R}}  K_{t_n, \, \left( n-1 \right)}^{\left(\alpha\right)}\left(p,z_{n-1}\right) \left[\frac{e^{-i(z_{n-1} - q)x_0}}{\sqrt{2\pi}} - \tilde{W_{\epsilon}}\left(z_{n-1} - q\right)\right] e^{it_{n}\left(z_{n-1}^{2} - q^2\right)} dz_{n-1}dt_{n}} dq & \leq \\ \leq 2M \int_{\mathbb{R}} \, \tilde{K}_{t, \, \left( n \right)}^{\left(\alpha\right)}\left(p,q\right) dq \leq 2M \left[ \underset{p \in \mathbb{R}}{\sup} \int_{\mathbb{R}} \;  \tilde{K}_{t, \, \left( n \right)}^{\left(\alpha\right)}\left(p,q\right) dq \right] & <  \infty,
\end{align*}

\noindent the estimate holding because of \cite{07_HS} thm. 3.4. Hence, concerning (\ref{eq: Integrale Delta 01}),
\begin{align*}
    \underset{\epsilon\downarrow0}{\lim} \, \int_{\mathbb{R}} \abs{\int_{0}^{t} \int_{\mathbb{R}}  K_{t_n, \, \left( n-1 \right)}^{\left(\alpha\right)}\left(p,z_{n-1}\right) \left[ \frac{e^{-i(z_{n-1} - q)x_0}}{\sqrt{2\pi}} - \tilde{W_{\epsilon}}\left(z_{n-1} - q\right) \right] e^{it_{n}\left( z_{n-1}^{2} - q^2\right)} dz_{n-1}dt_{n}} dq & \leq \\ \leq \int_{\mathbb{R}} \int_{0}^{t} \int_{\mathbb{R}} \abs{K_{t_n, \, \left( n-1 \right)}^{\left(\alpha\right)}\left(p,z_{n-1}\right)} \left[ \underset{\epsilon\downarrow0}{\lim} \abs{ \frac{ e^{-i( z_{n-1} - q)x_0}}{\sqrt{2\pi}} - \tilde{W_{\epsilon}}\left( z_{n-1} - q\right)}\right] dz_{n-1} dt_{n} dq & =0.
\end{align*}

\noindent On the other hand, regarding (\ref{eq: Integrale Delta 02}), the inductive hypothesis gives
\begin{align*}
    \underset{\epsilon\downarrow0}{\lim} \, \int_{\mathbb{R}} \left| \int_{0}^{t} \int_{\mathbb{R}} \left[ K_{t_n, \, \left( n-1 \right)}\left(p,z_{n-1}\right) - K_{t_n, \, \left( n-1 \right)}^{\left(\epsilon\right)}\left(p,z_{n-1}\right) \right] \tilde{W_{\epsilon}}\left(z_{n-1} - q\right) e^{it_{n}\left( z_{n-1}^{2} - q^2\right)} dz_{n-1}dt_{n} \right| dq & \leq \\ \leq  M \int_{\mathbb{R}} \int_{0}^{t} \left\{ \underset{\epsilon\downarrow0}{\lim} \int_{\mathbb{R}} \left| K_{t_n, \, \left( n-1 \right)}\left(p,z_{n-1}\right) - K_{t_n, \, \left( n-1 \right)}^{\left(\epsilon\right)}\left(p,z_{n-1}\right) \right| dq \right\} dt_{n}dz_{n-1} & = 0.
\end{align*}

\noindent By proceeding analogously for the adjoint relations, the Schur test gives
\begin{equation*}
    \norm{K_{t, \, \left( n \right)}^{\left(\alpha\right)} - K_{t, \, \left( n \right)}^{\left(\epsilon\right)}} \underset{ \epsilon \downarrow 0}{\longrightarrow}0, \; \forall t \in \mathbb{R}, \; \forall n \in \mathbb{N}.
\end{equation*}
\end{proof}

\begin{theorem}
    \justifying
    What follows holds.
    \begin{enumerate}
        \item $U_{\alpha}^{\ast}\left(t\right) \pi_S\left(a\right) U_{\alpha}\left(t\right) \in \pi_S\left[ \mathcal{R}\left( \mathbb{R}^2,\sigma \right) \right]$ for all $a \in \mathcal{R}\left( \mathbb{R}^2,\sigma \right)$;
        \item Denoted by $\left( H_\alpha, \mathcal{D}_{H_{\alpha}} \right)$ the self-adjoint operator on $L^2\left( \mathbb{R} \right)$ generating to the one parameter family of strongly continuous unitary operators $\left\{U_{\alpha}\left(t\right)\right\}_{t \in \mathbb{R}}$, $\left( H_\alpha, \mathcal{D}_{H_{\alpha}} \right)$ is affiliated to $\mathcal{R}\left(\mathbb{R}^2,\sigma \right)$;
        \item The map $\alpha_t: \, a \in \mathcal{R}(\mathbb{R}^2,\sigma) \longmapsto \alpha_t\left( a \right) \in \mathcal{R}(\mathbb{R}^2,\sigma)$, with
        \begin{equation*}
            \alpha_t \left(a\right) \doteq \pi_S^{-1} \left[ e^{itH_\alpha} \pi_S(a) e^{-itH_\alpha} \right],
        \end{equation*}
    \noindent results in an automorphism of $\mathcal{R}(\mathbb{R}^2,\sigma)$ for all $t \in \mathbb{R}$.
    \end{enumerate}
\end{theorem}

\begin{proof}
    \justifying
    \begin{enumerate}
        \item Directly from Proposition \ref{One Fixed-Center Result} and remark \ref{Clue Remark 1P}.
        \item It is a very well known fact that norm dynamical convergence\footnote{\begin{definition}
        \justifying
        Let $\mathcal{H}$ be a complex Hilbert space. Given self-adjoint operators $\left( A_n, \mathcal{D}_{A_n} \right), \, \left( A, \mathcal{D}_{A} \right)$, $A_n$ is \textbf{norm dynamically convergent} to $A$ if and only if, for all $t \in \mathbb{R}$, $\left\{e^{itA_n}\right\}_{n}$ converges to $e^{itA}$ with respect to the $\mathcal{B}\left( \mathcal{H} \right)$ norm. 
    \end{definition}
    $\hfill \square$} implies\footnote{See \cite{08_deO}, thm. 10.1.16.} norm resolvent convergence, therefore
    \begin{equation*}
        \norm{ U_{\alpha}\left( t \right) - e^{-itH_{\epsilon}} } \underset{\epsilon \downarrow 0}{\longrightarrow} 0, \, \forall t \in \mathbb{R}  \implies \norm{ \left( H_{\alpha} - i \lambda \mathds{1} \right)^{-1} - \left( H_{\epsilon} - i \lambda \mathds{1} \right)^{-1} } \underset{\epsilon \downarrow 0}{\longrightarrow} 0, \, \forall \lambda \in \mathbb{R}\setminus\{0\}.
    \end{equation*}
    \noindent Proposition \ref{Prop. 2.5} states that $\left( H_{\epsilon} - i \lambda \mathds{1} \right)^{-1} \in \pi_S\left[ \mathcal{R}\left( \mathbb{R}^2,\sigma \right) \right]$ for all $\lambda \in \mathbb{R}\setminus\{0\}$ and $\epsilon >0$. Since $\pi_S\left[ \mathcal{R}\left( \mathbb{R}^2,\sigma \right)\right]$ is closed with respect to the uniform norm topology, the affiliation of $\left( H_{\alpha}, \mathcal{D}_{H_{\alpha}} \right)$ results.
    \item Given $t \in \mathbb{R}$, the map
    \begin{equation}\label{Represented Heisenberg Dynamics}
        a \in \pi_S \left[\mathcal{R} \left( \mathbb{R}^2,\sigma \right)\right] \longmapsto e^{itH_\alpha} a e^{-itH_\alpha} \in \pi_S \left[\mathcal{R} \left( \mathbb{R}^2,\sigma \right)\right]
    \end{equation}
    is surely injective, being isometric. On the other hand, given $b \in \pi_S\left[\mathcal{R}\left(\mathbb{R}^2,\sigma \right)\right]$, because of proposition \ref{One Fixed-Center Result}, $e^{-itH_{\alpha}} \, b \, e^{itH_{\alpha}} \equiv d \in \pi_S\left[\mathcal{R}\left( \mathbb{R}^2,\sigma \right)\right]$, hence
    \begin{equation*}
        e^{itH_\alpha} \, d \, e^{-itH_{\alpha}} = b,
    \end{equation*}
    allowing to conclude that (\ref{Represented Heisenberg Dynamics}) is surjective. The same map is obviously a homomorphism; finally, being $\pi_S: \, \mathcal{R}\left(\mathbb{R}^2,\sigma \right) \longrightarrow \pi_S \left[ \mathcal{R}\left(\mathbb{R}^2,\sigma\right)\right] $ an isomorphism, the result follows.
    \end{enumerate}
\end{proof}


\subsection{Many Fixed-Centers Point Interactions}

\subsubsection{Finitely Many Fixed-Centers Point Interactions}

Focus is set on the symbolic Hamiltonian
\begin{equation}\label{N Centers PI Formal Hamiltonian}
    H = H_0 + \sum_{i=1}^{N} \, \alpha_i \, \delta\left( x - x_i \right),
\end{equation}

\noindent with $N \in \mathbb{N}$, coupling constants $\alpha_i \in \mathbb{R}\setminus \{0\}$ and fixed-centers location $x_i \in \mathbb{R}: \, x_i \neq x_j$. By setting $(\alpha) \equiv (\alpha_1,\ldots,\alpha_N) \in \mathbb{R}^N$ and $V = \sum_{i=1}^N \, \alpha_i \delta\left(\cdot - x_i\right)$,

\begin{enumerate}
    \item 
    \begin{equation*}
        \Tilde{V}(p) = \int_\mathbb{R} \left[\sum_{m=1}^N \alpha_m \delta(x-x_m)\right] e^{-ipx} \frac{dx}{\sqrt{2\pi}} = \sum_{m=1}^N \frac{\alpha_m}{\sqrt{2\pi}} e^{-ipx_m},
    \end{equation*}
    \noindent i.e. $\Tilde{V} \in L^\infty(\mathbb{R})$ and
    \item
    \begin{equation*}
        \Tilde{V}(u) = \sum_{m=1}^N \frac{\alpha_m}{\sqrt{2\pi}} e^{-ipx_m} = \sum_{m=1}^N \frac{\alpha_m}{\sqrt{2\pi}} e^{i(-p)x_m} = \overline{\sum_{m=1}^N \frac{\alpha_m}{\sqrt{2\pi}} e^{-i(-p)x_m}} = \overline{\Tilde{V}(-p)},
    \end{equation*}
    \vspace{2mm}
    \noindent for all $p \in \mathbb{R}$.
\end{enumerate}

\noindent Consequently, one legitimately relies on

\begin{align}
    K_{t, \, \left( 1 \right)}^{(\alpha)}(p,q) & = \left[\frac{e^{it(p^2-q^2)} - 1}{i(p^2-q^2)}\right] \left[ \sum_{m=1}^N \frac{\alpha_m}{\sqrt{2\pi}} e^{-i(p-q)x_m} \right] \equiv \sum_{m=1}^N K_{t, \, \left( 1 \right)}^{(\alpha), \, m}(p,q), \\
    K_{t, \, \left( n \right)}^{(\alpha)}(p,q) & = \int_0^t \int_\mathbb{R}  K_{t_n, \, \left( n-1 \right)}^{(\alpha)}(p,z_{n-1}) \left[\sum_{m=1}^N \frac{\alpha_m}{\sqrt{2\pi}}e^{-i( z_{n-1} - q)x_m}\right] e^{it_n(z_{n-1}^2 - q^2)} dz_{n-1}dt_n = \\
    & \equiv \sum_{m=1}^N K_{t, \, \left( n \right)}^{(\alpha), \, m}(p,q),
\end{align}

\noindent for all $ t,p,q \in \mathbb{R}$, to build  $\Gamma_{\left(\alpha\right)} \left( t \right)$ up, hence the unitary time evolution operator $U_{\left( \alpha \right)} \left( t \right), \; t \in \mathbb{R}$.

\begin{theorem}\label{Many Fixed-Center Result}
    \justifying
    Given $N \in \mathbb{N}$, let $\alpha_1, \ldots, \alpha_N \in \mathbb{R} \setminus \{0\}$ and non-negative smooth functions of compact support $W_1,\ldots,W_N \in C^{\infty}_{c}\left( \mathbb{R} \right)$ such that $\int_{\mathbb{R}} \, W_i = 1$ be. Considered the Schr{\"o}dinger Hamiltonians
    \begin{equation*}
        H_{\epsilon} = H_{0} + \sum_{i=1}^{N} \, \alpha_i W_{\epsilon,i} \equiv H_{0} + W_\epsilon, \quad \epsilon > 0,
    \end{equation*}
    \noindent where $W_{\epsilon,i}\left(x\right) = \epsilon^{-1} W_i\left(x/\epsilon \right), \, x \in \mathbb{R}, i \in \{1,\ldots,N\}$, for all $t \in \mathbb{R}$,
    \begin{equation*}
        \norm{U_{(\alpha)}\left(t\right) - e^{-itH_\epsilon}}_{\mathcal{B}\left(L^{2}\left(\mathbb{R}\right)\right)} \underset{\epsilon \downarrow 0}{\longrightarrow} 0.
    \end{equation*}
\end{theorem}

\begin{proof}
    By using
    \begin{align}
    K_{t, \, \left( 1 \right)}^{(\epsilon)}(p,q) & = \left[\frac{e^{it(p^2-q^2)} - 1}{i(p^2-q^2)} \right] \frac{\Tilde{W}_\epsilon(p-q)}{\sqrt{2\pi}} \equiv \sum_{m=1}^N K_{t, \, \left( 1 \right)}^{(\epsilon), \, m}(p,q) \\
    K_{t, \, \left( n \right)}^{(\epsilon)}(p,q) & = \int_0^t \int_\mathbb{R} K_{t_n, \, \left( n-1 \right)}^{(\epsilon)}(p,z_{n-1}) \frac{\Tilde{W}_\epsilon ( z_{n-1} - q )}{\sqrt{2\pi}} e^{it_n(z_{n-1}^2 - q^2)} \, dz_{n-1}dt_n = \sum_{m=1}^N K_{t, \, \left( n \right)}^{(\epsilon), \, m}(p,q)
\end{align}

\noindent for all $ t,p,q \in \mathbb{R}$, to build $\exp{\left(-itH_\epsilon \right)}$ up, $t \in \mathbb{R}$, one then has

\begin{align*}
    & \norm{ U_{\left( \alpha \right)}\left(t\right) - e^{-itH_\epsilon}} \leq \underset{n\in \mathbb{N}}{\sum} \norm{K_{n,t,s}^{(\alpha)} - K_{n,t,s}^{(\epsilon)}} \leq (\text{by Schur test}) \\ & \leq \underset{n\in \mathbb{N}}{\sum} \, \left\{\left[\underset{p}{\sup} \int_\mathbb{R}\abs{K_{t, \, \left( n \right)}^{(\alpha)}(p,q) - K_{t, \, \left( n \right)}^{(\epsilon)}(p,q)}dq \right] \left[\underset{p}{\sup} \int_\mathbb{R} \abs{K_{t, \, \left( n \right)}^{(\alpha),\ast}(p,q) - K_{t, \, \left( n \right)}^{(\epsilon),\ast}(p,q)} dq\right]\right\}^{\frac{1}{2}}.
\end{align*}

\noindent Therefore, by observing that

\begin{align}
    & \int_\mathbb{R} \abs{K_{t, \, \left( n \right)}^{(\alpha)}(p,q) - K_{t, \, \left( n \right)}^{(\epsilon)}(p,q)}dq = \int_\mathbb{R} \abs{\underset{m=1}{\sum^N}\left[K_{t, \, \left( n \right)}^{(\alpha), \, m}(p,q) - K_{t, \, \left( n \right)}^{(\epsilon), \, m}(p,q)\right]} dq \leq \\
    & \leq \underset{m=1}{\sum^N} \int_\mathbb{R} \abs{K_{t, \, \left( n \right)}^{(\alpha), \, m}(p,q) - K_{t, \, \left( n \right)}^{(\epsilon), \, m}(p,q)}dq, \label{Clue Estimate}
\end{align}

\noindent the result is proved as in Proposition \ref{One Fixed-Center Result}.
\end{proof}

\begin{theorem}
    \justifying
    What follows holds.
    \begin{enumerate}
        \item $U_{\left(\alpha\right)}^{\ast}\left(t\right) \pi_S\left(a\right) U_{\left(\alpha\right)}\left(t\right) \in \pi_S\left[ \mathcal{R}\left( \mathbb{R}^2,\sigma \right) \right]$ for all $a \in \mathcal{R}\left( \mathbb{R}^2,\sigma \right)$;
        \item Denoted by $\left( H_{\left(\alpha\right)}, \mathcal{D}_{H_{\left(\alpha\right)}} \right)$ the self-adjoint operator on $L^2\left( \mathbb{R} \right)$ generating the one parameter family of strongly continuous unitary operators $\left\{U_{\left(\alpha\right)}\left(t\right)\right\}_{t \in \mathbb{R}}$, $\left( H_{\left(\alpha\right)}, \mathcal{D}_{H_{\left(\alpha\right)}} \right)$ is affiliated to $\mathcal{R}\left(\mathbb{R}^2,\sigma \right)$;
        \item The map $\alpha_t: \, a \in \mathcal{R}(\mathbb{R}^2,\sigma) \longmapsto \alpha_t\left( a \right) \in \mathcal{R}(\mathbb{R}^2,\sigma)$, with
        \begin{equation*}
            \alpha_t \left(a\right) \doteq \pi_S^{-1} \left[ e^{itH_{\left(\alpha\right)}} \pi_S(a) e^{-itH_{\left(\alpha\right)}} \right],
        \end{equation*}
    \noindent results in an automorphism of $\mathcal{R}(\mathbb{R}^2,\sigma)$ for all $t \in \mathbb{R}$.
    \end{enumerate}
\end{theorem}

\begin{proof}
    \justifying
    The proof closely mimics the one of Proposition 3.2.
\end{proof}

\subsubsection{Countably Many Fixed-Centers Point Interactions}
\justifying

Given $\left\{ \alpha_i \right\}_{i \in \mathbb{N}} \in l^{1}\left( \mathbb{N} \right) \setminus \left\{ 0 \right\}$ the symbolic Hamiltonian 
    \begin{equation*}
        H = -\frac{d^2}{dx^2} + \sum_{i = 1}^{\infty} \, \alpha_i \delta\left( x - x_i \right),
    \end{equation*}
\noindent with $\left\{ x_i \right\}_{i \in \mathbb{N}} \subset \mathbb{R}$ such that $x_i \neq x_j$, for all $i \neq j$, is finally considered. Set $V = \sum_{i = 1}^{\infty} \, \alpha_i \delta\left( \cdot - x_i \right)$,
    \begin{enumerate}
        \item $\tilde{V}\left( p \right) = \frac{1}{\sqrt{2\pi}} \sum_{m \in \mathbb{N}} \, \alpha_m e^{-ipx_m} = \overline{\frac{1}{\sqrt{2\pi}} \sum_{m \in \mathbb{N}} \, \alpha_m e^{ipx_m}} = \overline{\tilde{V}\left( - p\right)}, \quad \forall p \in \mathbb{R}$ and
        \item $\tilde{V} \in L^{\infty}\left( \mathbb{R} \right)$,
    \end{enumerate}
\noindent therefore $\Gamma_{ \{\alpha_i\} } \left( t \right)$ can be obtained via (\ref{Kernel 1}), (\ref{Kernel n}).

\begin{theorem}\label{Infinite number of Deltas}
    \justifying
    What follows holds.
    \begin{enumerate}
        \item For all $t \in \mathbb{R}$, the unitary time evolution operator $U_{\{\alpha_i\}}\left(t\right) \doteq e^{-itH_0}\Gamma_{\{\alpha_i\}}\left(t\right)$ belongs to $\pi_S\left[ \mathcal{R}\left( \mathbb{R}^2,\sigma \right) \right]$. Moreover, denoted by $\left( H_{\left\{ \alpha_i \right\}}, \mathcal{D}_{H_{\left\{ \alpha_i \right\}}} \right)$ the self-adjoint operator on $L^2\left( \mathbb{R} \right)$ generating the one parameter family of unitary operators $\left\{ U_{ \left\{\alpha_i\right\} }\left( t \right) \right\}_{t \in \mathbb{R}}$, it is affiliated to $\mathcal{R}\left( \mathbb{R}^2,\sigma \right)$.
        \item $U_{\{\alpha_i\}}\left(t\right)^{\ast} \pi_S\left(a\right) U_{\{\alpha_i\}}\left(t\right) \in \pi_S\left[\mathcal{R}\left(\mathbb{R}^2,\sigma\right)\right]$ for all $a \in \mathcal{R}\left( \mathbb{R}^2,\sigma \right), \; t \in \mathbb{R}$.
        \item The map
        \begin{equation*}
            \alpha_t: \, a \in \mathcal{R}\left( \mathbb{R}^2,\sigma \right) \longmapsto \alpha_t\left(a\right) = \pi_S^{-1} \left[ U_{\left\{\alpha_i\right\}}\left( t \right)^{\ast} \pi_S(a) U_{\left\{\alpha_i\right\}}\left( t \right) \right] \in \mathcal{R} \left( \mathbb{R}^2,\sigma \right)
        \end{equation*}
        is an automorphism of $\mathcal{R}\left( \mathbb{R}^2,\sigma \right)$.
    \end{enumerate}
\end{theorem}

\begin{proof}
    \justifying
    Concerning 1., the result follows from \cite{09_HK} prop. 2, Proposition \ref{Many Fixed-Center Result} and the fact that norm dynamical convergence implies norm resolvent convergence. 2. and 3. are proved as in Proposition 3.2.
\end{proof}

\begin{remark}
    \justifying
    The non-trivial ideal structure of the resolvent algebra has already proved to be fundamental for the possibility of accommodating non-trivial quantum dynamics. The same feature is also of primary importance for the following final result to hold. $\hfill \square$
\end{remark}

\begin{theorem}
    Let $\mathfrak{K}_0$ be the $\text{C}^{\ast}-$subalgebra of $\pi_S\left[ \mathcal{R}\left( \mathbb{R}^2,\sigma \right) \right]$ generated by $\mathcal{B}_{\infty} \left( L^2\left( \mathbb{R} \right) \right)$ and the identity operator. $\left( \mathcal{K}_0 \equiv \pi_S^{-1}\left(\mathfrak{K}_0\right), \, \mathbb{R}, \, \beta \right)$, where
        \begin{equation*}
            \beta: \, t \in \mathbb{R} \longmapsto \beta_{t} \in \text{Aut}\left( \mathcal{K}_0 \right)
        \end{equation*}
    \noindent and
        \begin{equation*}
            \beta_{t}: \, a \in \mathcal{K}_0 \longmapsto \beta_{t}\left( a \right) \doteq \pi_S^{-1}\left[U(t)^{\ast} \, \pi_S\left(a\right) \, U(t)\right] \in \mathcal{K}_0,
        \end{equation*}
    \noindent $U(t) \in \mathcal{B}\left( L^2\left( \mathbb{R} \right) \right)$ propagating fixed point interactions, is a $\text{C}^{\ast}-$dynamical system.
\end{theorem}

\begin{proof}
    \justifying
    First of all, it is observed that Proposition 2.4 allows for $\mathcal{B}_{\infty} \left( L^2\left( \mathbb{R} \right) \right)$ to be contained in $\pi_S\left[ \mathcal{R}\left( \mathbb{R}^2,\sigma \right) \right]$; then, for all $t_0 \in \mathbb{R}$, $\norm{ U(t)^{\ast} U(t) - U(t_0)^{\ast}U(t_0)} \underset{t\rightarrow t_0}{\longrightarrow} 0$. On the other hand, given $\psi,\varphi \in L^2\left( \mathbb{R} \right)$, let the finite rank operator $T = \langle \psi, \cdot \rangle \varphi$ be. Fixed again $t_0 \in \mathbb{R}$, it results
    \begin{equation*}
        \norm{ U\left( t \right)^{\ast} T - U\left( t_0 \right)^{\ast} T} \leq \norm{\psi} \norm{U\left(t\right)^{\ast}\varphi - U\left(t_0\right)^{\ast}\varphi} \underset{t \rightarrow t_0}{\longrightarrow} 0.
    \end{equation*}
    Analogously,
    \begin{equation*}
        \norm{ TU\left( t \right) - TU\left( t_0 \right)} \leq \norm{\varphi} \norm{U\left(t\right)^{\ast}\psi - U\left(t_0\right)^{\ast}\psi} \underset{t \rightarrow t_0}{\longrightarrow} 0,
    \end{equation*}
    \noindent therefore
    \begin{align*}
        \norm{ U\left( t \right)^{\ast} T U\left( t \right) - U\left( t_0 \right)^{\ast} T U\left( t_0 \right) } \leq \norm{ U\left( t \right)^{\ast} T - U\left( t_0 \right)^{\ast} T } + \norm{ TU\left( t \right) - TU\left( t_0 \right) } \underset{t \rightarrow t_0}{\longrightarrow} 0.
    \end{align*}
    \noindent Linearity, density and continuity arguments prove the statement.
\end{proof}


\section{Conclusions}

\justifying
This paper shows that resolvent algebras can accommodate dynamics induced by self-adjoint Hamiltonians describing a single non-relativistic spinless particle undergoing one up to countably many different fixed point interactions located on the real line. Such a result is a one of a kind result, since, apart from Buchholz investigations, it deals with the open problem of establishing which physical systems the resolvent algebra formalism can manage. Moreover, the non-simple nature of such an algebra contributes in singling out a $\text{C}^{\ast}-$subalgebra of $\pi_S\left[ \mathcal{R}\left( \mathbb{R}^2,\sigma \right) \right]$ constituting a $\text{C}^{\ast}-$dynamical system.

\section*{Acknowledgements}
\justifying
 A. M. is grateful to professors Romeo Brunetti and Detlev Buchholz for precious comments, ideas and suggestions, allowing for this work to exist as it is as well as to the colleagues Andrea Moro, Matteo Crispino, Daniele Volpe and Dario De Stefano for their help with linguistics and debugging.

\end{document}